\documentclass[10pt,conference]{IEEEtran}
\usepackage{spconf}
\usepackage{enumerate}
\usepackage{cite}
\usepackage{amsmath}
\usepackage{amscd}
\usepackage{amssymb}
\usepackage{latexsym}
\usepackage{array}
\usepackage{tabularx}
\usepackage[ruled]{algorithm}
\usepackage{algorithmic}
\usepackage{epsfig,subfigure}

\usepackage{helvet}

\usepackage{multirow}
\usepackage{bigstrut}
\usepackage{array}
\usepackage{epstopdf}
\usepackage{graphicx}
\usepackage{float}

\newcommand{\PreserveBackslash}[1]{\let\temp=\\#1\let\\=\temp}
\newcolumntype{C}[1]{>{\PreserveBackslash\centering}p{#1}}
\newcolumntype{R}[1]{>{\PreserveBackslash\raggedleft}p{#1}}
\newcolumntype{L}[1]{>{\PreserveBackslash\raggedright}p{#1}}

\newtheorem{ther}{Theorem}
\newtheorem{deft}{Definition}

\newtheorem{prop}{Proposition}

\begin{document}
%
\title{Piecewise Toeplitz Matrices-based Sensing for Rank Minimization}
%
%
%


%
%
%

%

\name{Kezhi Li \qquad Cristian R. Rojas \qquad Saikat Chatterjee \qquad  H\r{a}kan Hjalmarsson}


\address{ACCESS Linnaeus Centre, Royal Institute of Technology (KTH) \\
   \{kezhi, crro, sach, hjalmars\}@kth.se }

\maketitle

\maketitle
\begin{abstract}
This paper proposes a set of piecewise Toeplitz matrices as the linear mapping/sensing operator $\mathcal{A}: \mathbf{R}^{n_1 \times n_2} \rightarrow \mathbf{R}^M$ for recovering low rank matrices from few measurements. We prove that such operators efficiently encode the information so there exists a unique reconstruction matrix under mild assumptions. This work provides a significant extension of the compressed sensing and rank minimization theory, and it achieves a tradeoff between reducing the memory required for storing the sampling operator from $\mathcal{O}(n_1n_2M)$ to $\mathcal{O}(\max(n_1,n_2)M)$ but at the expense of increasing the number of measurements by $r$. Simulation results show that the proposed operator can recover low rank matrices efficiently with a reconstruction performance close to the cases of using random unstructured operators.

\end{abstract}


\section{Introduction}

As a dual of the compressed sensing (CS) problem \cite{donoho-cs}, the matrix rank minimization problem has been extensively studied in recent years \cite{RECHT-Guaranteed,Candes-RobustPCA,Markovsky-StructuredLowRank}. This problem arises in many fields such as system identification \cite{Markovsky-StructuredLowRank}, computer vision \cite{Candes-RobustPCA} and quantum state tomography \cite{Gross-QuantumState}, where notions of order, dimensionality or complexity can be expressed in terms of the rank of a matrix. Let $\mathbf{X}$ be an $n_1 \times n_2$ matrix of low rank rank
 $r$, the sampling operation can be expressed as
 \begin{equation}\label{eq:y=Avec(X)}
   \mathbf{y}= \mathcal{A}(\mathbf{X})= \left[ \langle \mathbf{A}_1, \mathbf{X} \rangle , \cdots, \langle \mathbf{A}_M, \mathbf{X} \rangle  \right]^T = \mathbf{A} \text{vec}(\mathbf{X}),
 \end{equation}
 where $\mathbf{y}$ is a vector of measurements, $\mathcal{A}: \mathbf{R}^{n_1 \times n_2} \rightarrow \mathbf{R}^M$ is a linear transformation, $\mathbf{A}$ denotes the transform in matrix format in which $\mathbf{A} \in \mathbf{R}^{M \times (n_1 \cdot n_2)}$ (without loss of generality we assume that $n_1<n_2$); and $ \langle \mathbf{A}_j, \mathbf{X} \rangle = \text{Tr}(\mathbf{A}_j^T \mathbf{X} ), j \in \{1,\cdots M\}$, ($\text{vec}$) is the operator that vectorizes the matrix $\mathbf{X}$ by concatenating the columns as a long vertical vector $\mathbf{x} \in \mathbf{R}^{N}, N=n_1 n_2$.

 In compressed sensing, besides conventional Gaussian/Bernoulli random sensing approaches, many structured/deterministic sensing matrices have been proven suitable for recovery of compressible signals. The interest in using structured sensing matrices in CS stems from the application needs, mainly due to their low complexity in computation and memory, as well as hardware implementation \cite{Duarte-Structured,CS-Toeplitz-algorithm-Yin,Kezhi-Convolutional-SP}.

 While we note several attempts in CS to use structured sensing matrices, contributions to low rank matrix reconstruction are more scarce.  There are some efforts such as in \cite{RECHT-Guaranteed, Forbes-OnIdentity,Gross-RecoveringLow-Rank}. An existing key condition of $\mathcal{A}$ is the so called rank restricted isometry property (r-RIP) \cite{RECHT-Guaranteed}. However, examining the r-RIP of a given operator $\mathcal{A}$ is NP-hard.  In this paper, we show that if the $\mathbf{A}$ consists of independent identically distributed (i.i.d.) random variables forming a piecewise Toeplitz structure, it is feasible to recover the objective low rank matrix from its measurements uniquely. Instead of studying its r-RIP, we convert the uniqueness problem of rank minimization into a compressed sensing problem, and analyze the performance bounds of proposed matrices by using the tools from structured sensing matrix analysis in CS. The extension is not trivial, since the vectorized low rank matrix is no longer sparse. By utilizing such technique we may reduce the memory required to store the sampling operator from $\mathcal{O}(n_1n_2M)$ to $\mathcal{O}(M\max(n_1,n_2))$ at the expense of a few measurements under mild assumptions.

The rest of the paper is organized as follows. In Section II we formulate the problem and introduce the proposed piecewise Toeplitz matrices. The main result and the corresponding proof are presented in Section III. Simulations are given in Section IV and finally Section V addresses the conclusion.

\subsection{Relations to Previous Works}
There are existing random or structured matrices for the rank minimization problem, such as \cite{RECHT-Guaranteed, Forbes-OnIdentity, Gross-RecoveringLow-Rank}. The idea of this paper different from the previous work is to use Toeplitz structured operators inspired from structured sensing matrices in CS \cite{Duarte-Structured,Haupt-Toeplitz}. After decomposing the low rank matrix, we analyze the uniqueness of recovering a block sparse vector (as defined in \cite{Eldar-Block-Spa}). In \cite{Haupt-Toeplitz} the authors also adopted the Gershgorin circle theorem to bound the eigenvalues of Toeplitz matrices, while in this paper the different coherence expression makes the analysis much more complicated. The idea of decomposing low rank matrix by its columns also relates to CUR matrix decomposition and the Nystrom method \cite{Mahoney-CUR}. In contrast, to these previous approaches, the core problem here is the unique reconstruction rather than the decomposition, so the coherence and RIP analyses are adopted.

\section{Piecewise Toeplitz matrix-based Sensing}

\subsection{Problem Formulation}

To exploit the low rank property, we suppose that the low rank matrix $\mathbf{X}=[\mathbf{x}_1, \mathbf{x}_2, \cdots, \mathbf{x}_{n_2}]$ has rank $r$ so that $r$ of its columns $\mathbf{x}_i$ can represent the remaining $(n_2-r)$ columns explicitly by their linear combination. Denote the selected $r$ columns as $\mathbf{x}^{\diamond}$, $\mathbf{x}^{\diamond} \subset \{\mathbf{x}_1, \mathbf{x}_2, \cdots, \mathbf{x}_{n_2} \}$, ${\diamond} \subset \{1,2, \cdots, n\}$, $\text{card}({\diamond}) = r$; and the remaining $(n_2-r)$ columns as $\mathbf{x}^{*}$, $* \subset \{1, \cdots, n_2\}$, $\text{card}(*) = n_2-r$. We call them primary columns and secondary columns, respectively. Then the $(n_2-r)$ secondary columns can be represented as
\begin{equation}
    \mathbf{x}_i = \sum_{j={1}}^{r} \alpha_{ij}\mathbf{x}_{{\diamond}_j},
\end{equation}
$i \in \{*_1, *_2, \cdots, *_{n_2-r} \}$. Its matrix multiplication form is \begin{footnotesize}
\begin{equation}\label{eq:*=ad}
\begin{array}{c} \left[ \begin{array}{c} \mathbf{x}_{*_1} \\ \mathbf{x}_{*_2} \\ \vdots \\ \mathbf{x}_{*_{n_2-r}} \end{array} \right] = \left\{\underbrace{ \left[ \begin{array}{ccc} \alpha_{{*_{1}}\diamond_{1}}  & \cdots & \alpha_{*_{1}\diamond_{r}}
\\  \alpha_{*_{2}\diamond_{1}}  & \cdots & \alpha_{*_{2}\diamond_{r}} \\ \vdots & \ddots & \vdots \\  \alpha_{*_{n_2-r}\diamond_{1}}  & \cdots & \alpha_{*_{n_2-r}\diamond_{r}} \end{array} \right]  } \otimes \mathbf{I}_{n_1} \right\} \left[ \begin{array}{c} \mathbf{x}_{\diamond_1} \\ \mathbf{x}_{\diamond_2} \\ \vdots \\ \mathbf{x}_{\diamond_{r}} \end{array} \right],\\ { \boldsymbol\alpha}  \end{array}
\end{equation} \end{footnotesize}
where $\otimes$ denotes the Kronecker product. Normally $r \ll n_2$, i.e. $\boldsymbol\alpha$ is a tall matrix. Thus $\mathbf{x}=\text{vec}(\mathbf{X})$ can also be decomposed in a block sparse manner as
\begin{equation}\label{eq:x=Psif}
\mathbf{x} = \mathbf{\Psi} \cdot \mathbf{f}
\end{equation}
where $\mathbf{\Psi}$ is a sparse matrix with block diagonal matrices $\alpha_{ij} \mathbf{I}_{n_1}$ or $\mathbf{I}_{n_1}$ of size $n_1 \times n_1$, $\mathbf{f}$ is a block $r$-sparse vector (adopting the definition in \cite{Eldar-Block-Spa}) $\mathbf{f} = \left[ \mathbf{0}, \cdots, \mathbf{x}_{\diamond_1}^T, \mathbf{0},  \cdots ,  \mathbf{x}_{\diamond_r}^T , \cdots \right]^T$.
Then (\ref{eq:y=Avec(X)}) can be written as:
\begin{equation}\label{eq:y=A Psi f}
   \mathbf{y}= \mathcal{A} (\mathbf{X}) = \mathbf{A} \mathbf{x} = \mathbf{A} \mathbf{\Psi} \mathbf{f} = \mathbf{\Theta} \mathbf{f},
\end{equation}
where ${\bf \Theta}={{\bf A}}{\bf \Psi}$. This block sparse problem has been studied by analyzing the block-restricted isometry constant \cite{Eldar-Block-Spa}. However, here the matrix $\mathbf{\Theta}$ is formed by the multiplication of $\mathbf{A}$ and an unknown structured sparsifying matrix $\mathbf{\Psi}$ in (\ref{eq:x=Psif}). And obviously this $\mathbf{\Psi}$ is neither unitary, nor can be constructed delicately from complete bases. It is an unknown block diagonal sparse matrix determined by the low rank matrix $\mathbf{X}$. The same low rank matrix may even lead to multiple $\mathbf{\Psi}$ and $\mathbf{f}$. Thus the conventional theory of block CS may encounter difficulties for such a problem.

\subsection{Piecewise Toeplitz Matrices}\label{sec:PTM}

\begin{deft}[\bf Piecewise Toeplitz]
A set of matrices $\{\mathbf{A}_1, \mathbf{A}_2, \cdots, \mathbf{A}_M\}$ of size $n_1 \times n_2$ are defined as piecewise Toeplitz matrices if
\begin{equation}\label{eq:a1=[a11a12]}
\begin{split}
\mathbf{a}_1= \text{vec}(\mathbf{A}_1)^T&= [\mathbf{a}_{11}^T,\mathbf{a}_{12}^T, \cdots, \mathbf{a}_{1n_2}^T], \\
&\vdots \\
\mathbf{a}_M= \text{vec}(\mathbf{A}_M)^T&= [\mathbf{a}_{M1}^T,\mathbf{a}_{M2}^T, \cdots, \mathbf{a}_{Mn_2}^T ],
\end{split}
\end{equation}
where $\mathbf{a}_{ij}$ denotes the $j$th column of the matrix $\mathbf{A}_i$, and their piecewise concatenation matrices
\begin{equation}\label{eq:A[i]}
\begin{array}{cccc}
\mathbf{A}[1] =\left[ \begin{array}{c}
\mathbf{a}_{11}^T \\
\vdots \\
\mathbf{a}_{M1}^T \\
\end{array} \right] &%
\cdots &
\mathbf{A}[n_2] =\left[ \begin{array}{c}
\mathbf{a}_{1n_2}^T \\
\vdots \\
\mathbf{a}_{Mn_2}^T \\
\end{array} \right]

\end{array}
\end{equation}
are all Toeplitz.
\end{deft}

\begin{prop}\label{prop:Theta expression}
For the measurement process $\mathbf{y}(j)= \langle \mathbf{A}_j, \mathbf{X} \rangle, j\in \{1, \cdots, M\}$, and a matrix $\mathbf{X}$ of size $n_1 \times n_2 $ and rank $r$, (\ref{eq:y=Avec(X)}) is equivalent to $\mathbf{y} = \mathbf{\Theta f} $, where $\mathbf{f}$ is a block $r$-sparse vector, and $ \mathbf{\Theta} $ of size $M \times (n_1n_2)$ has the structure $\mathbf{\Theta} =\left[ \begin{array}{cccccc}
\mathbf{\Theta}[1] &
\mathbf{\Theta}[2] &
\cdots &
\mathbf{\Theta}[i] &
\cdots &
\mathbf{\Theta}[n]
\end{array} \right]$ in which
\begin{equation}\label{eq:Theta0}
\mathbf{\Theta}[i]= \left\{ \begin{array}{ccl}
\mathbf{A}[i]+ \sum_{*}{\alpha_{*i}\mathbf{A}[*]} & \mbox{if} & i \in \{\diamond\}
\\ \mathbf{0} & \mbox{if} & i \notin \{\diamond\}
\end{array}\right.
\end{equation}
for all $i \in \{1, \cdots, n_2\}$. $\{\mathbf{A}[i]\}$ are matrices derived from $\{\mathbf{A}_j\}$ by concatenating their columns piecewisely as in (\ref{eq:A[i]}). $\{\diamond\} , \{*\}$ represent the sets of primary columns and secondary columns indexes, respectively. $\text{card}(\diamond)=r, \text{card}(*)=n_2-r$. $\mathbf{\Theta}[i]$ are also Toeplitz when $i \in \{\diamond\}$.
\end{prop}
\begin{proof}
Because $\mathbf{X}$ has rank $r$,
\begin{equation}\label{eq:y=AAAPsif}
\mathbf{y} = \left[ \mathbf{A}[1] \ \mathbf{A}[2] \ \cdots \  \mathbf{A}[n] \right]  \mathbf{\Psi}
 \mathbf{f} = \mathbf{\Theta f}.
\end{equation}
Then it is straightforward to verify the expression in \emph{Prop.} 1. Because $\mathbf{A}[i]$ are Toeplitz, $\mathbf{\Theta}[i]$ must be Toeplitz as well when $i \in \{\diamond\}$. 
\end{proof}

\textbf{Remark}: Since the decomposition $\mathbf{x}= \mathbf{\Psi f}$ is not unique, $\mathbf{\Theta}$ have different expressions corresponding to various $\mathbf{f}$, which distinguishes (\ref{eq:y=AAAPsif}) from CS with multiple solutions $\mathbf{f}$. Fortunately what we need to recover is not $\mathbf{f}$ but $\mathbf{x}$. Multiple $\mathbf{f}$ may lead to a unique solution $\mathbf{x}$.

\section{Main Result and Proof}

\begin{prop}[\bf Unique Recovery]\label{prop:unique-rec}
The reconstruction of matrix $\mathbf{X}$ with size $n_1 \times n_2$ and rank $r$ in (\ref{eq:y=A Psi f}) has a unique solution $\hat{\mathbf{X}}$ if $\mathbf{\Theta f} \neq 0$ holds for every $\mathbf{f} \neq 0$ in (\ref{eq:y=A Psi f}) which is block $2r$-sparse, where $\mathbf{\Theta }$ has the structure in (\ref{eq:Theta0}).
\end{prop}
\begin{proof}
%
%
%
Assume that there is a new solution $\mathbf{X}^*$ to (\ref{eq:y=A Psi f}) with $\text{rank}(\mathbf{X}^*) \leq r$, $\mathbf{X}^* \neq \hat{\mathbf{X}}$, which means that
\begin{equation}
\hat{\mathbf{x}}= \hat{\mathbf{\Psi}} \hat{\mathbf{f}},  \mathbf{x}^*= \mathbf{\Psi}^* \mathbf{f}^*, \hat{\mathbf{x}} \neq \mathbf{x}^*.
 \end{equation}
Let $\mathbf{X}' = \mathbf{X}^*-\hat{\mathbf{X}}$, where $\mathbf{X}'$ is a nonzero matrix of rank at most $2r$. Then there must exist an $\mathbf{f}'$ that is block $2r$-sparse such that $\mathbf{ A \Psi}' \mathbf{f}' =\mathbf{0}$, where $\mathbf{\Psi}' \mathbf{f}' = \mathbf{X}'$, which contradicts the assumption. Please note that $\mathbf{f}' \neq \mathbf{f}^*-\hat{\mathbf{f}}$ if $ \mathbf{\Psi}^*\neq \hat{\mathbf{\Psi}} $. The result is arrived from the above analyses.
\end{proof}
\textbf{Remark}: \emph{Prop.} \ref{prop:unique-rec} is an extension of the uniqueness guarantee from a sparse vector in CS to a low rank matrix. $\mathbf{\Theta}$ depends on the unknown matrix $\mathbf{X}$, $\mathbf{f}$ is not unique. However, because $\mathbf{X}'$ has rank at most $2r$, it must be decomposed into $\mathbf{\Psi}'$ and a $2r$-sparse vector $\mathbf{f}'$, which is in contradiction to the assumption.

\begin{prop}[\bf $\epsilon$ Bound]\label{prop:epsilon_bound}
Consider $\mathbf{\Theta}= \mathbf{A \Psi} $ with the structure in (\ref{eq:Theta0}). Denote by $\mathbf{\Theta}[\diamond]$ the submatrix formed by retaining the column blocks of $\mathbf{\Theta}$ indexed by $\diamond$. If the normalized Gram matrix $\mathbf{G}$ of every $\mathbf{\Theta}[\diamond]$ has bound
\begin{equation}
\begin{split}
|g_{ii}-1| \leq \epsilon_1,\\
 \mathcal{R}_i =\sum_{ \substack{ j=1 \\ j \neq i} }{|g_{ij}|} \leq \epsilon_2,
 \end{split}
 \end{equation}
 with some positive values $0<\epsilon_1, \epsilon_2<1$, then the eigenvalues of $\mathbf{G}(\mathbf{\Theta}[\diamond])$ are bounded by $(1-\epsilon, 1+\epsilon)$, $\epsilon=\epsilon_1+ \epsilon_2$, and (\ref{eq:y=A Psi f}) has the unique block $2r$-sparse solution $\hat{\mathbf{f}}$ when $\epsilon$ is bounded by the RIP constant. 
\end{prop}
\begin{proof}
The proof is based on the Gershgorin circle theorem \cite{Horn-Matrix}, and can be derived from the RIP and \emph{Prop.} \ref{prop:unique-rec}.
\end{proof}

Now we give the main result of this paper. The details of \textit{Condition 1} and \textit{Assumption 1} are provided in the Appendix.

\begin{ther}[\bf Main Result]
Consider the measurements $\mathbf{y}(j)= \langle \mathbf{A}_j, \mathbf{X} \rangle, j\in \{1, \cdots, M\}$. Let $\mathbf{A}_j$ be piecewise random Toeplitz matrices whose entries satisfy \textit{Condition 1}. $\mathbf{X}$ is a rank $r$ matrix satisfying \textit{Assumption 1}, when $r^2<\mathcal{O}(n_1)$, $M \geq \mathcal{O}\left(r^2(n_1 +n_2) \log(n_1n_2)\right)$, then there exists a constant $c>0$ such that for any fixed $\mathbf{X}$ has a unique solution $\hat{\mathbf{X}}=\mathbf{X}$ with probability exceeding $1- \exp{\left(  - \frac{c M  }{ r^2 n_1^2} \right)}$ as $n_2 =n_1^2$.
\end{ther}

\begin{proof}
The proof exploits inner products of any two columns of $\mathbf{\Theta}$ in order to bound the eigenvalues. Denote the $q$th column in matrix $\mathbf{A}[p]$ as $\mathbf{A}[p,q], p \in \{1,\cdots, n_2\}, q \in \{1,\cdots, n_1\}$. For the entries in one row of the Gram matrix $\mathbf{G}$, there are four circumstances of non-zero $\mathbf{\theta}[p_1, q_1]^T \mathbf{\theta}[p_2 ,q_2]$: (1) $p_1 =p_2, q_1 =q_2 $; (2) $p_1 \neq p_2, q_1 =q_2 $; (3)$p_1 =p_2, q_1\neq q_2 $; (4) $p_1 \neq p_2, q_1\neq q_2 $. We will analyze them case by case.

(1) When $p_1 =p_2, q_1 =q_2 $, they are the diagonal entries of the Gram matrix. Without loss of generality, we calculate $\mathbf{\theta}_{\diamond_i q}^T \mathbf{\theta}_{\diamond_i q}, i \in \{1, \cdots, r\}, q \in \{1, \cdots, n_1\}$,  which implies $p_1= \diamond_i, q_1 = q_2=q$.
 Following the expressions in \emph{Prop.} \ref{prop:Theta expression}, the $((\diamond_i-1) \cdot n_1 +q)$th column of $\mathbf{\Theta}$ can be calculated is
 \begin{equation}
 \mathbf{\theta}[\diamond_i, q]= \mathbf{A}[\diamond_i, q] + \sum_{*}{\alpha_{* \diamond_i} \mathbf{A}[*, q]}.
\end{equation}
Suppose $|\mathbf{\theta}[\diamond_i, q](k)| \leq a$ where $a$ is a positive bound. There is some positive value $\gamma_{i}$ that $\mathbb{E}(\theta[\diamond_i, q]^2(k)) = \gamma_{i}^2 \sigma^2$ for every $k \in \{1, \cdots, M\}$, $\sigma^2=1/M$. By exploiting Hoeffding's inequality we have
\begin{equation}
\text{Pr}\left\{ \left| \sum_{k=1}^M{\theta[\diamond_i, q]^2(k)} - \gamma_i^2 \sigma^2 M \right| \geq t_0 \right\}   \leq 2 \exp{\left( -\frac{2 t_0^2 }{M a^4} \right) }.
\end{equation}

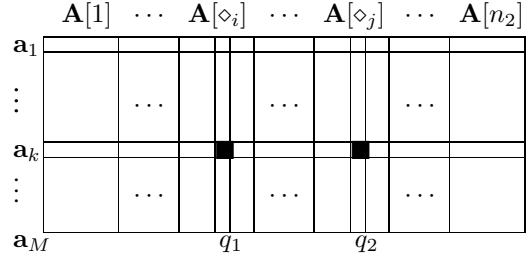
\begin{figure}[t]\label{fig:matrix1}
   \centering
\setlength{\unitlength}{2mm}
\begin{picture}(20,17)(8,0)
\linethickness{0.075mm}

\multiput(2,1)(5,0){2}%
{\line(0,1){13}}
\multiput(11,1)(5,0){2}%
{\line(0,1){13}}
\multiput(20,1)(5,0){2}%
{\line(0,1){13}}
\multiput(29,1)(5,0){2}
{\line(0,1){13}}

\put(2,1){\line(1,0){32.01}}
\put(2,6){\line(1,0){32.01}}
\put(2,7){\line(1,0){32.01}}
\put(2,13){\line(1,0){32.01}}
\put(2,14){\line(1,0){32.01}}

\put(0,0){$\mathbf{a}_M$}
\put(0,13){$\mathbf{a}_1$}
\put(0,3){$\vdots$}
\put(0,6){$\mathbf{a}_k$}
\put(0,9){$\vdots$}

\put(8,15){$\cdots$}
\put(17,15){$\cdots$}
\put(26,15){$\cdots$}
\put(8,9){$\cdots$}
\put(17,9){$\cdots$}
\put(26,9){$\cdots$}
\put(8,3){$\cdots$}
\put(17,3){$\cdots$}
\put(26,3){$\cdots$}

\put(3.2,15){$\mathbf{A}[1]$}
\put(11.6,15){$\mathbf{A}[\diamond_{i}]$}
\put(20.6,15){$\mathbf{A}[\diamond_{j}]$}
\put(29.6,15){$\mathbf{A}[n_2]$}

\multiput(13.4,1)(1,0){2}%
{\line(0,1){13}}
\multiput(22.4,1)(1,0){2}%
{\line(0,1){13}}
\put(13.7,0){$q_1$}
\put(22.7,0){$q_2$}

\put(0,9){$\vdots$}


\put(13.4,5.9){$\blacksquare$}
\put(22.4,5.9){$\blacksquare$}

\end{picture}
   \centering
  \caption{The schematic diagram of matrix $\mathbf{A}$. $\mathbf{a}_k, \mathbf{A}[i]$ are defined in \textit{Deft. 1}. $q_1$ and $q_2$ are internal column indexes. The two entries with black squares are denoted as $\mathbf{A}[\diamond_{i},q_1](k), \mathbf{A}[\diamond_{j},q_2](k)$, respectively.
  }\label{fig:matrix1}
\end{figure}

(2) When $p_1 \neq p_2, q_1 = q_2 $, they represent two columns in different block $\mathbf{\Theta}[p_1], \mathbf{\Theta}[p_2]$ but have relatively the same internal positions. Let $p_1 = \diamond_i, p_2 = \diamond_j, q_1=q_2= q$, the two columns $\mathbf{\theta}[\diamond_i, q], \mathbf{\theta}[\diamond_j, q]= \mathbf{A}[\diamond_j, q]$ can be expressed as
 \begin{equation}
 \begin{split}
 \mathbf{\theta}[\diamond_i, q]= \mathbf{A}[\diamond_i, q] + \sum_{*}{\alpha_{* \diamond_i} \mathbf{A}[*, q]} \\
 \mathbf{\theta}[\diamond_j, q]= \mathbf{A}[\diamond_j, q] + \sum_{*}{\alpha_{* \diamond_j} \mathbf{A}[*, q]}.
 \end{split}
\end{equation}
Because all entries in $\mathbf{A}$ are i.i.d. in different blocks,
\begin{equation}
\begin{split}
& \mathbb{E}{\left( \mathbf{\theta}[\diamond_i, q](k) \cdot \mathbf{\theta}[\diamond_j, q](k) \right) } \\
&=\mathbb{E}{\left( \sum_{*}{\alpha_{* \diamond_i} \mathbf{A}[*, q](k)} \sum_{*}{\alpha_{* \diamond_j} \mathbf{A}[*, q](k)} \right) } \\
&= \sum_{*}{\alpha_{* \diamond_i} \alpha_{* \diamond_j} \cdot \sigma^2}
\end{split}
\end{equation}
Although $\mathbf{\theta}[\diamond_i, q](k),  \mathbf{\theta}[\diamond_j, q](k)$ are dependent due to the mutual combination terms, we still can use Hoeffding's inequality to bound the summation, because for different $k$, $\mathbf{\theta}[\diamond_i, q](k) \cdot  \mathbf{\theta}[\diamond_j ,q](k)$ are i.i.d.. Let $\kappa_{ij}=  \sum_{*}{\alpha_{* \diamond_i} \alpha_{* \diamond_j}} $, then there exists some positive $t_1$ such that
\begin{equation}
\text{Pr}{\left\{ \left| \mathbf{\theta}[\diamond_i, q]^T \mathbf{\theta}[\diamond_j ,q] - \kappa_{ij}\sigma^2 M \right| \geq t_1 \right\} } \leq 2\exp{\left( - \frac{t_1^2}{2M a^4} \right)}.
\end{equation}


(3) When $p_1 = p_2, q_1 \neq q_2 $, the two columns are in the same block $p_1 = p_2= \diamond_i, i\in\{1, \cdots r\}$ with different internal index $q_1, q_2$:
 \begin{equation}
 \begin{split}
 \mathbf{\theta}[\diamond_i, q_1]= \mathbf{A}[\diamond_i ,q_1] + \sum_{*}{\alpha_{* \diamond_i} \mathbf{A}[* ,q_1]} \\
 \mathbf{\theta}[\diamond_i ,q_2]= \mathbf{A}[\diamond_i, q_2] + \sum_{*}{\alpha_{* \diamond_i} \mathbf{A}[*, q_2]}.
 \end{split}
\end{equation}
Here a natural problem comes out: for these two columns, they are not independent any more due to the Toeplitz structure. Here we use ``divide and conquer" technique that separates the sum into two groups that have no mutual terms. For instance, if $q_2> q_1, q_2-q_1 = d$, the sum can be divided as
\begin{equation}
 \begin{split}
\theta[\diamond_i, q_1]^T \mathbf{\theta}[\diamond_i, q_2]&=\begin{array}{c} \underbrace{\sum_{k=1}^{d}\theta[\diamond_i, q_1](k)\theta[\diamond_i, q_2](k+d) +\cdots} \\ \text{first group} \end{array}  \\
&\begin{array}{c} \underbrace{+\sum_{k=d+1}^{2d}\theta[\diamond_i, q_1](k)\theta[\diamond_i, q_2](k+d) + \cdots} \\ \text{second group} \end{array},
 \end{split}
\end{equation}
and it is always possible to find a partition that divides $\theta[\diamond_i, q_1]^T \mathbf{\theta}[\diamond_i, q_2]$ into two parts as sums with size $\{\frac{M}{2}, \frac{M}{2}\}$ for even $M$ and $\{\frac{M-1}{2}, \frac{M+1}{2}\}$ for odd $M$. Thus for some positive value $t_2$,
\begin{equation}
  \text{Pr}\left\{ \left| \theta[\diamond_i, q_1]^T \mathbf{\theta}[\diamond_i, q_2] \right| \geq t_2 \right\} \leq 4 \exp{\left( - \frac{t_2^2}{8 M a^4} \right)}.
\end{equation}

(4) When $p_1 \neq  p_2, q_1 \neq q_2 $,
similar to case (3), we divide $\mathbf{\theta}_{\diamond_i q_1}^T \mathbf{\theta}_{\diamond_j q_2}$ into $2$ parts without mutual terms, giving
\begin{equation}
  \text{Pr}\left\{ \left| \theta[\diamond_i, q_1]^T \mathbf{\theta}[\diamond_j, q_2] \right| \geq t_3 \right\} \leq 4 \exp{\left( - \frac{t_3^2}{8 M a^4} \right)}.
\end{equation}

Finally, we normalize the Gram matrix and summarize the diagonal elements and off-diagonal elements. Assume that $\mathbf{X}$ satisfies the \emph{statistical low rank property} defined in the \textit{Def. \ref{def:SLRP}} (see Appendix).
Suppose $a=\sqrt{c_0/M}\gamma_m$, $\gamma_m= \max(\gamma_i)$. Let $\epsilon_{21}= \epsilon_{21}'+\frac{\kappa_{ij}}{\gamma^2}, \epsilon_{21}'= \frac{t_1}{\gamma^2}$, $\gamma= \mathbb{E}(\gamma_i)$, and let $t_2 = \gamma^2 \epsilon_{22}, t_3= \gamma^2 \epsilon_{23}, \epsilon_{1}=\epsilon_{2m}= \frac{1}{4}\epsilon, \epsilon_2=\sum_{m}\epsilon_{2m}, m=\{1,2,3\}$. After tedious calculation of bounding bias factor $\gamma^2_i$ and $\kappa_{ij}$, we obtain the normalized result by summarizing the four cases above and adopting the Gershgorin theorem
\begin{equation}
\text{Pr}\left\{ \bigcup_{i=1}^{n_1n_2} \{ \sum_{j\neq i}^{r n_1}\left| g_{ij}   \right| \geq \epsilon_2  \}\right\} \leq 4 n_1^2 n_2^2 \exp{\left(  - \frac{ M \epsilon^2 }{128 c_0^2 r^2 n_1^2} \right)}.
\end{equation}
Hence there must exist a constant $0<c< \frac{\epsilon^2}{128 c_0^2}$ that
\begin{equation}
\text{Pr}\left\{  \text{Not Unique Recovery} \right\} \leq \exp{\left(  - \frac{c M  }{ r^2 n_1^2} \right)} \ \
\end{equation}
whenever
\begin{equation}
M  \geq \left( \frac{384 c_0^2}{\epsilon^2 -128c c_0^2} \right)r^2(n_1+n_2) \log(n_1 n_2),
\end{equation}
as $r^2<\mathcal{O}(n_1), n_2 =n_1^2, n_2 \rightarrow \infty$ then $\text{Pr} \rightarrow 0$. Use \emph{Prop.} \ref{prop:unique-rec},\ref{prop:epsilon_bound} to derive the last step, which completes the proof.
\end{proof}

\textbf{Remark}: 1) For Gaussian matrices with i.i.d. entries, the measurements for the recovery of a low rank matrix should be at least
$M \sim \mathcal{O} \left( r(n_1+n_2)\log(n_1 n_2) \right)$ \cite{RECHT-Guaranteed}. In our case, the measurement price is the extra factor $r$.

2) For random Gaussian matrices, one needs $\mathcal{O}(n_1n_2M)$ memory to store the operator. By using the piecewise Toeplitz structure, we are able to reduce the memory requirement to $\mathcal{O}\left((M+n_1)n_2\right)$ as $n_1<n_2<M$.

3) The proof holds when $r^2<\mathcal{O}(n_1), n_2 =n_1^2$ in order to avoid the situation that $M \geq n_1 n_2$. In practice it is feasible to apply $\mathbf{A}$ to the case when $n_1,n_2$ are close. Our simulations verify this conjecture numerically.

\section{Simulations}

Extensive simulations have been carried out to compare
the reconstruction performances of random and
proposed operators. Here we present some results.

We utilize $3$ different algorithms, including a) cvx toolbox to minimize the nuclear norm \cite{Boyd-ConvexOpt}; b) Alternating Least-Squares (ALS) algorithm \cite{Zachariah-ALSforLR} c) Directional-ALS algorithm \cite{Li-AlternatingStr}, to compare the reconstruction results, respectively. For each algorithm, we recover a $50 \times 50$ random low rank matrix using different $\mathbf{A}$ such as Gaussian/Bernoulli operators, 3-valued operators \cite{RECHT-Guaranteed} and finally random piecewise Toeplitz operators with truncated Gaussian entries. Fig. 1 depicts a comparison of reconstruction errors with increasing rank $r$ at sampling rate $\rho = M/{(n_1n_2)}=0.3$. Each point is recorded as an average of 200 trials. From these curves one can observe that the performance of the
proposed operator is close to that of random
matrices, which are typically considered as the optimal universal operators. In addition, the proposed operators may be equipped with fast reconstruction algorithms potentially by exploiting the Toeplitz structure like in CS \cite{CS-Toeplitz-algorithm-Yin}.

\section{Conclusion}
This paper proposes piecewise Toeplitz matrices as structured linear operators in the matrix minimization problem, and proves that it is feasible to recover the low rank matrix uniquely when the number of measurements exceeds $\mathcal{O}(r^2(n_1+n_2) \log(n_1 n_2))$ under mild assumptions. Experimental results show that
the proposed operators compare favorably with existing random operators.

\begin{figure}[!ht]
   \centering
   \begin{minipage}[ht]{1\linewidth}
   \centering
  \includegraphics[width=7.5cm]{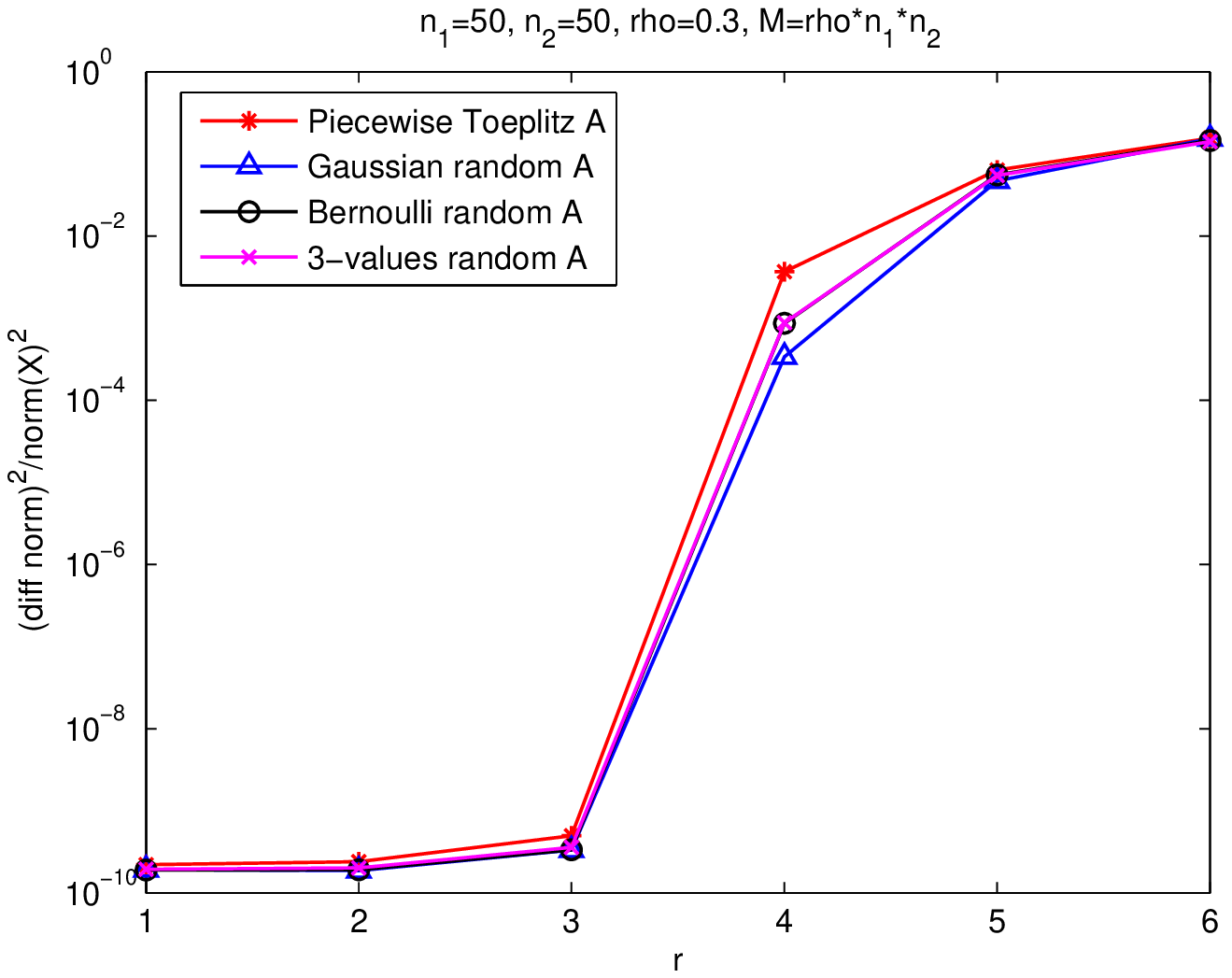} 
\centerline{(a)} 
   \end{minipage}
   \begin{minipage}[ht]{1\linewidth}
   \centering
   \includegraphics[width=7.5cm]{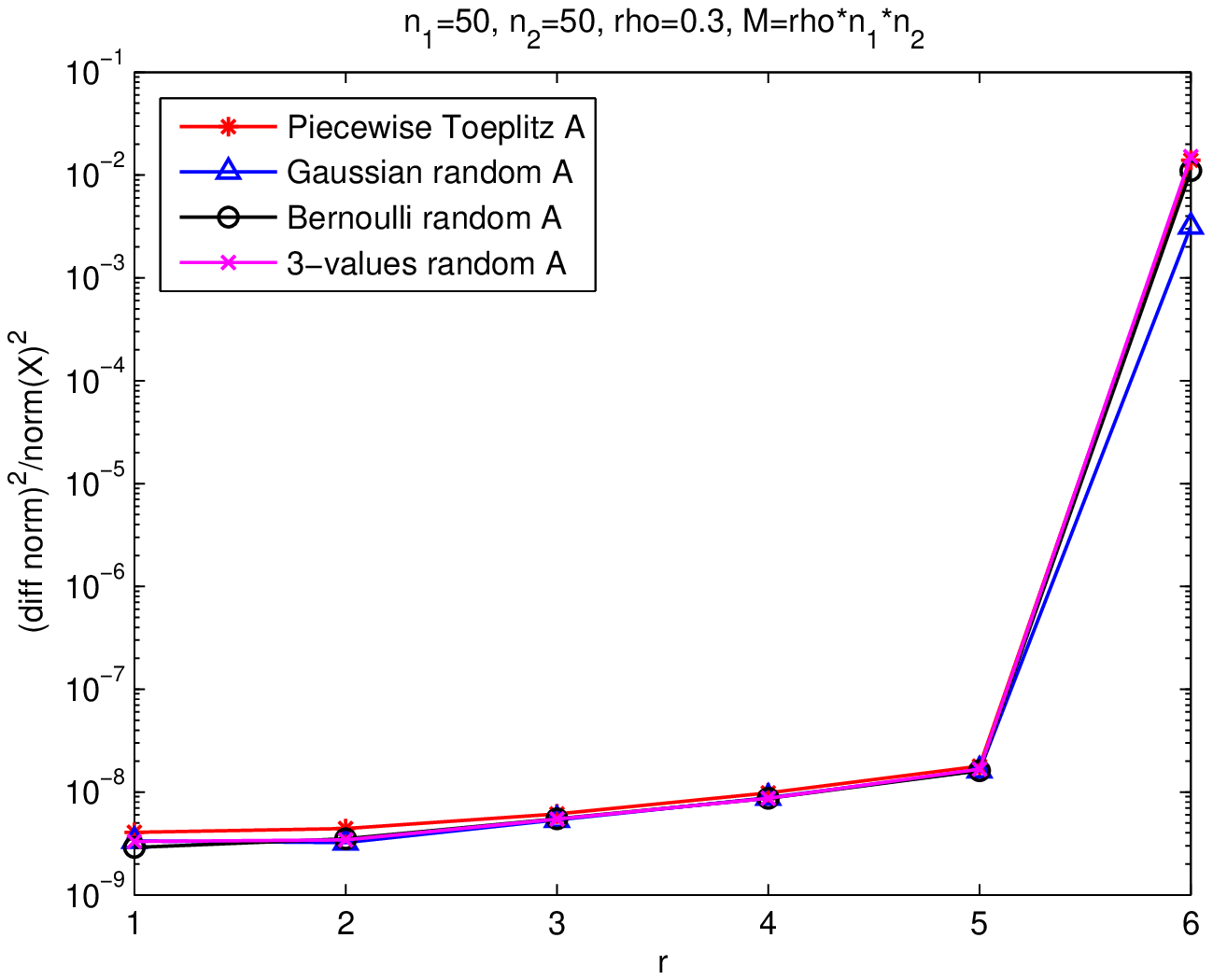} 
\centerline{(b)} 
   \end{minipage}
      \begin{minipage}[ht]{1\linewidth}
   \centering
   \includegraphics[width=7.5cm]{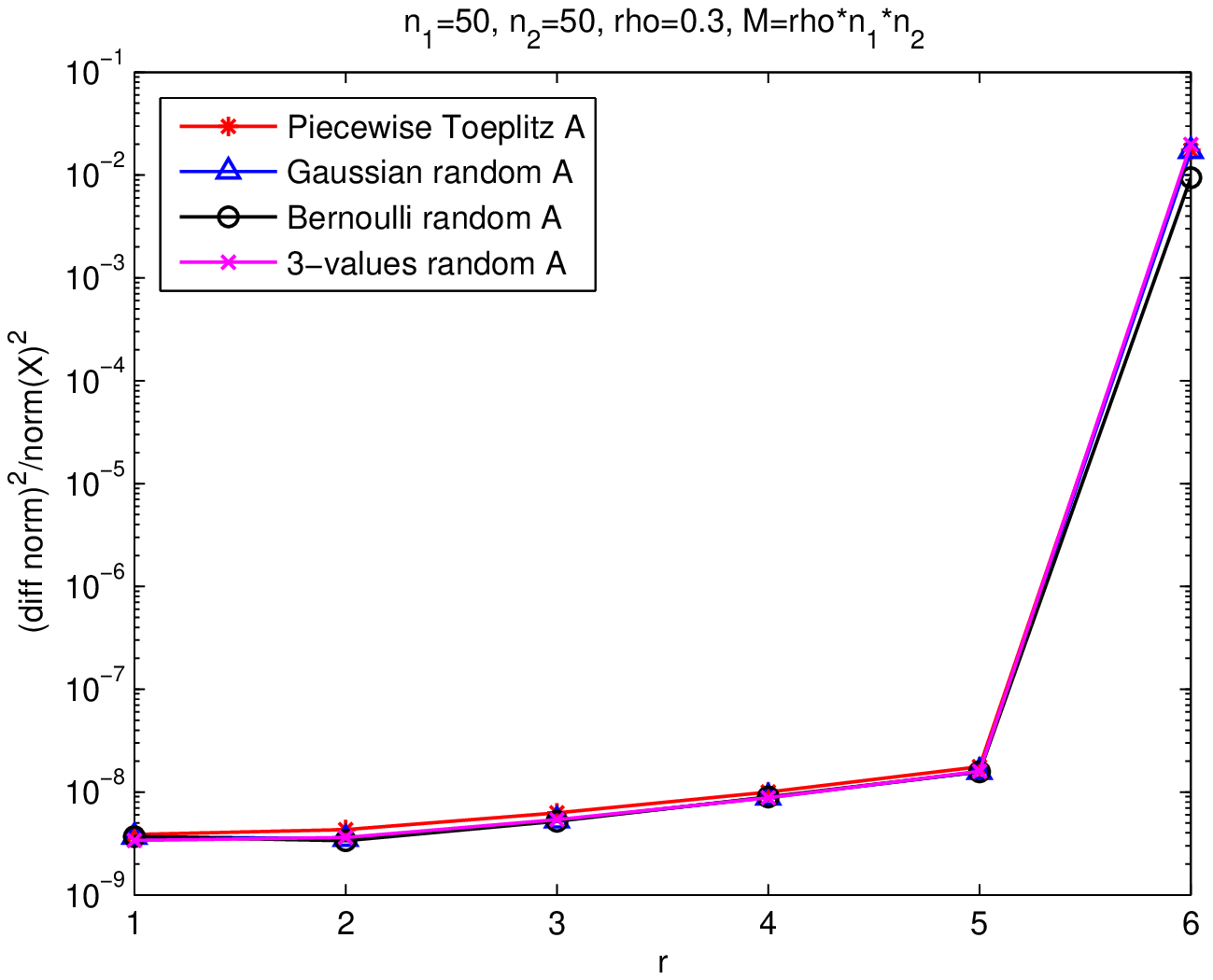} 
\centerline{(c)} 
   \end{minipage}
  \caption{Simulation results (a) using cvx toolbox to minimize the nuclear norm of $\mathbf{X}$  (b) using Alternating Least-Squares (ALS) algorithm (c) using Directional-ALS algorithm. } \vspace{-2mm}
  \end{figure}\label{fig:M-ALS_vs_ALS2}

\section{Appendix}
\textit{Condition 1} ($\mathbf{A}_j$) \\
$\mathbf{A}_j, j\in \{1,\dots, M\}$ is an $n_1 \times n_2$ matrix whose entries are bounded i.i.d. $0$-mean $1/M$-variance random variables satisfying $|\mathbf{A}_j(p,q)|\leq \sqrt{c_0/M}$ for some $c_0>1$, $p\in \{1, \cdots, n_1\}, q\in \{1, \cdots, n_2\}$; in addition, $\{\mathbf{A}_1, \cdots, \mathbf{A}_M\}$ are a set of piecewise Toeplitz matrices defined in Sec. \ref{sec:PTM}.

\textit{Assumption 1} ($\boldsymbol\alpha$) \\
$\boldsymbol\alpha$ is the parameter matrix of $r$-rank $\mathbf{X}$ defined in (\ref{eq:*=ad}). We assume that $\alpha_{{*_{p}}\diamond_{q}}, p \in \{1,\cdots, n-r\}, q \in \{1,\cdots, r\}$ satisfy the \emph{statistical low rank property} defined below.
\begin{deft}[\bf Statistical Low Rank Property]\label{def:SLRP}
A $r$-rank matrix $\mathbf{X}$ has the statistical low rank property if the $\mathbf{\alpha}_{ij}$ in (\ref{eq:*=ad}) are i.i.d. random variables with zero mean and variance $\sigma_{\alpha}^2$.
\end{deft}
In the proof of the main result we set $\sigma_{\alpha}^2=1/r$ to keep the variance of $\mathbf{x}_i$ identical due to (\ref{eq:*=ad}).




\end{document}